\documentclass[12pt]{article}
\usepackage{main}
\usepackage{fullpage}
\usepackage{authblk}

\renewcommand{\Pr}{\mathbb{P}}
\newcommand{\Exp}{\mathbb{E}}
\newcommand{\ExpCond}[2]{\Exp\left(#1 \mid #2\right)}

\newcommand{\inMIS}{\texttt{in-MIS}}
\newcommand{\notInMIS}{\texttt{out-MIS}}

\newcommand{\transmit}{\texttt{transmit}}
\newcommand{\listen}{\texttt{listen}}
\newcommand{\sleep}{\texttt{sleep}}

\newcommand{\undecided}{\texttt{undecided}}
\newcommand{\lose}{\texttt{lose}}
\newcommand{\win}{\texttt{win}}
\newcommand{\commit}{\texttt{commit}}
\newcommand{\terminate}{\texttt{Terminate}}

\newcommand{\Prob}[1]{\Pr\left(#1\right)}
\newcommand{\DeltaEst}{\Delta_{est}}
\newcommand{\peterAlg}{$\Call{LowDegreeMIS}{}$}
\newcommand{\RecBackoff}{Rec-EBackoff}
\newcommand{\SndBackoff}{Snd-EBackoff}

\DeclareMathOperator{\poly}{poly}

\title{Energy-Efficient Maximal Independent Sets in Radio Networks}

\date{}

\author[1]{Dominick Banasik}
\author[2]{Varsha Dani}
\author[3]{Fabien Dufoulon}
\author[4]{Aayush Gupta}
\author[5]{Thomas P. Hayes}
\author[6]{Gopal Pandurangan\thanks{Supported in part by  NSF grant CCF-2402837 and ARO grant W911NF-231-0191.}}
\affil[1,2]{\small Rochester Institute of Technology, Rochester, NY, USA\\
\texttt{db5335@rit.edu}, \texttt{varsha.dani@rit.edu}}
\affil[3]{\small Lancaster University, Lancaster, UK\\
\texttt{f.dufoulon@lancaster.ac.uk}}
\affil[4,6]{\small University of Houston, Houston, TX, USA\\
\texttt{agupta56@cougarnet.uh.edu}, \texttt{gopal@cs.uh.edu}}
\affil[5]{\small University at Buffalo, Buffalo, NY, USA\\
\texttt{thayes2@buffalo.edu}}

\begin{document}

\maketitle

\begin{abstract}
The maximal independent set (MIS) is one of the most fundamental problems in distributed computing, and it has been studied intensively for over four decades.
This paper focuses on the MIS problem in the \emph{radio network} model, a standard model widely used to model \emph{wireless} networks, particularly ad hoc wireless and sensor networks. 
Energy is a premium resource in these networks, which are typically battery-powered. Hence, designing distributed algorithms that use as little energy as possible is crucial. We use the well-established energy model where a node can be \emph{sleeping}
or \emph{awake} in a round, and only the awake rounds (when it can send or listen)  determine the \emph{energy complexity} of the algorithm, which we want to minimize.

We present new, more energy-efficient MIS algorithms in radio networks with {\em arbitrary} and unknown graph topology. We present algorithms for two popular variants of the radio model ---  with collision detection  (CD) and without collision detection (no-CD). Specifically, we obtain the following results:
\begin{enumerate}
    \item \emph{CD model}: We present a randomized distributed MIS algorithm with energy complexity $O(\log n)$, round complexity $O(\log^2 n)$, and failure probability $1 / \poly(n)$, where $n$ is the network size. We show that our energy complexity is optimal by showing a matching $\Omega(\log n)$ lower bound.
    \item \emph{no-CD model}: In the more challenging no-CD model, we present a randomized distributed MIS algorithm with energy complexity $O(\log^2n \log \log n)$, round complexity $O(\log^3 n \log \Delta)$, and failure probability $1 / \poly(n)$. The energy complexity of our algorithm is significantly lower than
    the round (and energy) complexity of $O(\log^3 n)$ of the best known distributed MIS algorithm of Davies [PODC 2023] for arbitrary graph topology.  
\end{enumerate}
\end{abstract}

\smallskip


\setcounter{page}{1}

\section{Introduction}

The Maximal Independent Set (MIS) problem is a fundamental problem in graph theory and distributed computing, with numerous applications in network design, resource allocation, and parallel computing. It is also one of the best-studied symmetry-breaking problems in distributed networks.  In particular, solving the MIS problem efficiently is crucial in radio networks due to the inherent challenges posed by \emph{wireless} communication, such as contention and collision. 
For example, in the popular ad hoc wireless and sensor networks, nodes are deployed with no infrastructure; in fact, nodes may not even know which nodes are close (i.e., neighbors). Unlike wired networks, nodes cannot broadcast at will to discover their neighbors; radio interference and collisions make it unlikely for such uncoordinated communications to reliably transmit any information.
To coordinate communication, one can first construct an MIS, then use it as a building block for setting up a \emph{communication backbone}.
Such communication applications of MIS have been studied extensively in ad hoc wireless and sensor networks (see e.g.,~\cite{wagner}).
However, this leads to the problem of \emph{first constructing an MIS} starting with no underlying knowledge of the neighborhood or topology in radio networks, which is the focus of this paper.

Another major issue in radio networks, such as ad hoc sensor networks, is that nodes are typically battery-powered and hence energy-constrained. Most of the energy consumption of the nodes is when they are transmitting or listening. On the other hand, very little energy is consumed when the nodes are {\em sleeping}, i.e., when the radio devices are switched off; in such a state, a node does not send or listen (but messages sent to it are lost)~\cite{energy1,podc2020}.  This necessitates the design of \emph{energy-efficient} distributed algorithms where nodes try to minimize the number of rounds they send or listen in.

In this paper, we focus on designing \emph{energy-efficient} distributed algorithms for MIS in radio networks with \emph{arbitrary (and unknown)} topology. Distributed algorithms (and lower bounds) for constructing MIS in radio networks have been studied extensively for many years (cf. Section \ref{sec:related}).  Almost all of the algorithms studied, except the recent algorithm of Davies~\cite{davies2023}, assumed the underlying graph to be of a particular type, such as unit disk~\cite{moscibroda-wattenhoffer-podc2005},  growth-bounded~\cite{wattenhofer-growthbounded}, or more generally, bounded independence~\cite{daum-podc13, daum-kuhn}. Furthermore, all these works focused on minimizing the \emph{round complexity}. 
In particular, the work of Davies~\cite{davies2023} --- which is most relevant to this work --- gave an $O(\log^3 n)$ round distributed MIS algorithm for radio networks that works for an arbitrary (and unknown) underlying topology (throughout, $n$ denotes the number of nodes in the network). This is the best known round-efficient MIS algorithm for radio networks that works
under an arbitrary topology. However, none of the above prior works focused on designing energy-efficient algorithms in radio networks, and  theier energy complexities 
can be {\em as high} as their respective round complexities. 

\subsection{Radio Network Model and Energy Complexity} \label{sec:model}
Radio networks are characterized by their broadcast communication model, where nodes communicate by transmitting messages over shared channels. In this model, a message sent by a node can be received by all its neighbors within the transmission range. However, if multiple nodes transmit simultaneously, \emph{collisions} occur, leading to communication failures. This necessitates the development of robust distributed algorithms that can effectively handle such collisions and ensure reliable communication.

We assume we have an underlying communication network modeled by an {\em arbitrary} graph $G$. Each node in $G$ has a transmitter
and receiver to communicate with other nodes. There is an edge between two nodes in $G$
if they are within each other's transmission range. We note that the graph $G$
is \emph{unknown} to all the nodes.   In particular, we will assume that nodes do not even know who their
neighbors are in the graph until they have explicitly heard from them during the algorithm. This is sometimes called the \emph{ad hoc} model~\cite{davies2023}.  

We assume that the nodes start out with shared knowledge of a value $n$, which is an upper bound on the number of nodes in the network;
our bounds require only an estimate of $n$ within a polynomial factor.\footnote{Although we informally think of $n$ as being equal to the actual number of nodes in $G$, this is not at all necessary; the only disadvantage
to having $n$ much larger than the actual number of nodes is a possible increase in the time and energy complexities, and message lengths, of our algorithms.}  Similarly, we shall assume that the nodes have shared knowledge of a value $\Delta$ which is an upper bound on the maximum degree of the graph.  In settings where $\Delta$ is not given, our algorithms may still be applied, using $n$ in place of $\Delta$; however, in this case, our energy and round complexity guarantees become worse.  A more sophisticated approach to use when $\Delta$ is not known, is to guess a series of increasing values for $\Delta$, running our algorithm for each guess.  When the guesses are too small, portions of the output may fail to be independent, in which case affected vertices must detect this fact, and repeat the algorithm with the next, larger, value for $\Delta$.  The details are sufficiently complicated that we omit them from this version of the paper, mentioning only that using $2^{2^i}$ as the $i$th guess for $\Delta$ seems to work well, and carries an $O(\log \log n)$ factor overhead for our energy complexity, and an $O(1)$ factor overhead for our round complexity.

We assume all nodes start in the same state (with possibly no predesignated IDs) but have access to private random bits. (This allows nodes to locally generate
$O(\log n)$-bit IDs which are unique with high probability.) Unless otherwise stated, we allow a failure probability of $\sfrac{1}{\poly(n)}$. 
We assume the standard RADIO-CONGEST model, which constrains the size of messages that can be sent in a single round, limiting them to $O(\log n)$ bits. This constraint reflects practical limitations in real-world networks, where bandwidth is limited, and large messages can cause congestion and delays. (The alternative model, where there is no such bandwidth constraint, is called the RADIO-LOCAL model.) 

We assume synchronous communication where time is divided into discrete \emph{rounds} (or timesteps). In each round, a processor can be in one of two states: \emph{awake} or \emph{sleeping}. In the awake state, a node can \emph{either}  \emph{transmit} or \emph{listen} (but not both due to the radio nature).
In the sleeping model~\cite{podc2020}, only the awake rounds are counted towards the \emph{energy complexity} (also called as
\emph{awake complexity} in some works --- cf. Section \ref{sec:related}). On the other hand, both sleeping and awake rounds
are counted towards the \emph{round complexity}.
We assume local computation (performed by a node in a round) is free, but all our algorithms do local computation at a very small cost (at most logarithmic in $n$). 

As in~\cite{davies2023,wattenhofer-growthbounded}, we assume \emph{synchronous wake-up}, i.e., all nodes wake up simultaneously and can execute the algorithm immediately. (We note that \emph{asynchronous wake-up} has also been studied in several prior works~\cite{moscibroda, jurdin, moscibroda-wattenhoffer-podc2005,daum-kuhn}).

A message broadcast by a node $u$  at time $t$ is received by a neighbor $v$ of $u$
 at round $t$ if:  (i) $v$ listens at time $t$ and (ii) no other neighbor of $v$  transmits at time $t$. 
If some other neighbor of $v$ also transmits at time $t$, $v$'s reception will depend on how collisions are handled.
Two standard and well-studied models of handling collisions are (i) with collision detection (CD) and (ii) without collision detection (no-CD) --- defined as follows.
\begin{itemize}
\item \textbf{Collision-Detection (CD) model:} In the CD  model, a listening node \emph{can} distinguish between silence (no neighbors sending) and a collision (more than one neighbor sending). Thus,
if more than one neighbor of $v$ transmits, then $v$ will hear a collision.
\item  \textbf{No Collision Detection (no-CD) model:} In the no-CD model a listening node \emph{cannot} distinguish between silence and collisions between \emph{two or more} messages.  Thus, if more than one neighbor
of $v$ transmits, then $v$ will hear nothing (i.e., silence).
\end{itemize}
 
The no-CD radio model is more challenging and has been used extensively in prior works (see e.g.,~\cite{energy4, yehuda89, davies2023, moscibroda-wattenhoffer-podc2005} and the references therein). We note that the CD model is also well-studied (see e.g.,~\cite{yehuda89,ghaffari-podc2013,wattenhofer-growthbounded} and the references therein) and closely related to the well-studied \emph{beeping model}~\cite{beeping}, see Section \ref{sec:related}. 
A no-CD algorithm will also work on the CD model, but may be less time- or energy-efficient than a respective CD algorithm. 

\subsection{Maximal Independent Set (MIS)}
 \emph{Maximal independent set} (\emph{MIS})  is one of the most well-studied problems in distributed graph algorithms.
Given a graph with $n$ nodes, each node must (irrevocably) commit to being in a subset $M \subseteq V$ (called the MIS) or not such that (i) every node is either in $M$ or has a neighbor in $M$ and (ii) no two nodes in $M$ are adjacent to each other.

The MIS problem has been studied extensively for the last four decades in several distributed computing models (see e.g.,~\cite{ghaffari-soda,podc2023,ghaffari-sleeping,ghaffari-podc2023,davies2023} and the references therein). 
In this paper, our focus is on algorithms for the \emph{radio network model} with low \emph{energy complexity}.

\subsection{Our Results}

We present new, more \emph{energy-efficient} MIS algorithms in radio networks on $n$ nodes with {\em arbitrary} and unknown graph topology.  Specifically, we obtain the following results:
\begin{enumerate}
    \item \textbf{Lower Bound:} We show a lower bound of $\Omega(\log n)$ on the energy complexity. This lower bound
    applies to both the CD and no-CD models. 
    \item \textbf{Energy-Optimal MIS in the CD model}: We present a randomized distributed MIS algorithm with energy complexity $O(\log n)$, round complexity $O(\log^2 n)$, and failure probability $\sfrac{1}{\poly(n)}$. Our algorithm is energy-optimal
    because of the above lower bound.
    \item \textbf{Energy-Efficient MIS in the no-CD model}: In the more restrictive no-CD model, we present a randomized distributed MIS algorithm with energy complexity $O(\log^2n \log \log n)$, round complexity $O(\log^3 n \log \Delta)$, and failure probability $\sfrac{1} {\poly(n)}$.
    The energy complexity of our algorithm is significantly lower than
    the round (and energy) complexity of $O(\log^3 n)$ of the best known distributed MIS algorithm of Davies~\cite{davies2023} for arbitrary graph topology.
    Furthermore, our energy complexity essentially matches (up to a $\log \log n$ factor) the best known lower bound for \emph{round complexity}  $\Omega(\log^2 n)$~\cite{davies2023, farach-colton} in the no-CD model.
\end{enumerate}

   Our algorithms perform only unary communication, i.e., nodes only transmit a ``1'' bit if they transmit at all. In particular,
   our CD algorithm will also work in the simpler {\em beeping} model with the same energy and round complexities (cf. Section \ref{sec:cddesc}).

Our algorithms are an energy-efficient implementation of a Luby-like algorithm~\cite{luby1985simple,metivier} for radio networks.
A somewhat straightforward implementation of Luby for radio networks will take $O(\log^2 n)$ energy and rounds in the CD model and $O(\log^4 n)$ energy and rounds in the no-CD model.
In particular, in the more challenging no-CD model, it is non-trivial to improve the round complexity to $O(\log^3 n)$ rounds as was done in~\cite{davies2023}; they do so by using
an efficient radio implementation of Ghaffari's algorithm~\cite{ghaffari-soda} for CONGEST (wired) networks. However, the algorithm of~\cite{davies2023} also has $O(\log^3 n)$ energy complexity, as some nodes can be awake for so many rounds.

Our approach uses several non-trivial tools to implement a Luby-like algorithm energy-efficiently in a radio network --- improving
the above $O(\log^3 n)$ bound (in the no-CD model) by almost a logarithmic factor for energy complexity. 

\subsection{Additional Related Work}\label{sec:related}

The literature on the MIS problem is vast. We focus mainly on those relevant to this work, i.e., \emph{radio network model} and \emph{energy-efficient} algorithms for \emph{MIS}.

As mentioned earlier, almost all of the prior works focused on improving the \emph{round complexity} of MIS in the radio model, and these focused on special families of graphs such as unit disk graphs~\cite{moscibroda} or bounded independence graphs~\cite{daum-podc13,daum-kuhn}. These algorithms typically assume the \emph{no-CD} model and asynchronous wake-up and run in $O(\log^2 n)$  rounds. This bound can be improved for \emph{multi-channel networks}~\cite{daum-kuhn}, but the $O(\log^2 n)$ bound applies for standard single-channel networks (as assumed in this paper and many others cited above). It can be shown that $\Omega(\log^2 n)$ is a lower bound for the round complexity~\cite{davies2023,farach-colton}. 
For arbitrary graph topology, the best known complexity bound is $O(\log^3 n)$~\cite{davies2023}; this result, like ours, assumes synchronous wake-up. As pointed out in~\cite{davies2023}, the $\Omega(\log^2 n)$ round complexity lower bound applies to synchronous wake-up as well. In the CD model, the work of~\cite{wattenhofer-growthbounded} showed a tight bound
of $\Theta(\log n)$ on the \emph{round complexity} of MIS  in
\emph{growth-bounded} graphs. Note that our tight $\Theta(\log n)$ energy bound in the CD model applies to arbitrary graphs.

There also has been extensive work on the so-called \emph{beeping} model~\cite{beeping, jeavons} where nodes can only communicate by beeping, which is equivalent to transmitting (or not transmitting) a bit (i.e., unary communication). Collision detection is assumed in the sense that if more than one neighbor of a node (say $v$) beeps, then $v$ hears at least one beep. This is similar to the CD radio model, except in two ways: (i) In the radio model a node can
send $O(\log n)$-sized message in a round, but a beep contains no information (except its presence).
(ii) In the beeping model, the best known MIS algorithms typically assume \emph{sender-side} collision detection, see e.g., the work
of~\cite{jeavons} which gives an optimal $O(\log n)$-round MIS algorithm in the beeping model. Sender-side collision detection means that the sender can detect a beep from its neighbors even when the sender is transmitting a beep. In the radio model, sender-side CD is not assumed --- a node can only send or receive in any round, if they do both, then they will not hear anything. 

The {\em sleeping} (or energy) model has been studied extensively in recent years.
As mentioned in Section \ref{sec:model}, the main feature of this model is that a node can be either in the sleeping or awake state 
in any round, and only rounds spent in the awake states are counted towards the energy (also called as \emph{awake}) complexity.
The energy model for \emph{radio networks} (also called SLEEPING-RADIO~\cite{podc2023}) used in this paper was introduced and studied in ~\cite{energy1,energy2}   (inspired by earlier work on energy-efficient algorithms in radio networks e.g.,~\cite{nakano,jurdin,pajak}).
Energy-efficient algorithms for radio networks for several
problems such as broadcast, leader election, breadth-first search, maximal matching, diameter and minimum-cut computation have been studied~\cite{energy1,energy2,energy3,energy4,energy5, CS22, CDJ23,Chang23}.

In another line of work, energy-efficient (or \emph{awake-efficient}) algorithms for MIS have been designed in the \emph{sleeping} model for CONGEST networks (called SLEEPING-CONGEST~\cite{podc2023}).
In this model, unlike radio networks (and more powerful), nodes can broadcast (or unicast) simultaneously without collisions; in other words, it is simply the standard CONGEST model with nodes having the additional flexibility to sleep to save energy. This model was introduced in~\cite{podc2020}. This paper showed that MIS (in general graphs) can be
solved in $O(1)$ \emph{node-averaged} awake complexity, which is measured by the \emph{average} number of rounds a node is
awake. This is in contrast to the worst-case awake (or energy) complexity (that is used in this paper and all other papers cited here) which is the
worst-case number of rounds a node is awake until it finishes the algorithm.
The
worst-case awake complexity of their MIS algorithm is $O(\log n)$, while the
worst-case complexity (that includes all rounds, sleeping and awake) is
$O(\log^{3.41}n)$ rounds. A  question left open in~\cite{podc2020} is whether one can design an MIS algorithm 
with  $o(\log n)$ worst-case awake complexity (even in the LOCAL model).
This question was answered in~\cite{podc2023} where it was shown that MIS can be solved in $O(\log \log n)$ awake complexity (with high probability) which is exponentially better than the round complexity lower bound of $\Omega(\sqrt{\log{n}/\log\log{n}})$.
Several recent works in the SLEEPING-CONGEST  model for fundamental problems such as MIS, approximate matching and vertex cover, spanning tree, minimum spanning tree, coloring, and other problems include~\cite{ghaffari-sleeping,BM21,AMP22,ghaffari-podc2023,ghaffari-podc2024, dufoulon-coloring,alkida}. 

\subsection{Organization}

The rest of the paper is organized as follows. We present our lower bound for the CD model in Section~\ref{subsec:CDMISLowerBound}. We then present our CD algorithm in Section~\ref{sec:CDMIS}. Section~\ref{subsec:noCDPrimitives} introduces preliminary techniques for our no-CD algorithm, which follows in Section~\ref{sec:noCDMIS}. We note that in Sections~\ref{sec:CDMIS}--\ref{sec:noCDMIS}, unless there is an explicit reference to an appendix, missing proofs can be found in the full version of the paper. We conclude and highlight open questions in Section~\ref{sec:conc}.

\section{ \texorpdfstring{An $\Omega(\log n)$}{A Logarithmic} Energy Complexity Lower Bound for MIS}
\label{subsec:CDMISLowerBound}

In this section, we show an $\Omega(\log n)$ lower bound on the energy complexity of MIS in the CD (and no-CD) model. Note that this lower bound is unconditional of the round complexity. 


\begin{theorem}
In radio networks with CD, any algorithm that solves MIS with probability strictly more than $e^{-1/4}$ has energy complexity of at least $1/2 \log n$.
\end{theorem}

\begin{proof}
Suppose to the contrary that every node is awake for $o(\log n)$ rounds and consider the following $n$-node anonymous graph $G$ (assume that $n$ is a multiple of 4): $G$ is the union of $n/4$ disjoint edges and $n/2$ isolated nodes. Intuitively, in this graph, each isolated node must join the independent set whereas each non-isolated node must agree with its neighbor about which of them joins the independent set. 





Consider a node $v$ that is awake for only $b$ rounds. If $v$ hears nothing in those rounds, then by Bayes' Law, 
the conditional probability that $v$ is an isolated node, given $v$'s state of knowledge, is at least $1/2$. Therefore
$v$ must decide to join the independent set. 
On the other hand, let $v$ be a non-isolated node and $w$ be its neighbor. 
It is necessary for at least one message sent by either $v$ or $w$ to be heard by the other, since otherwise, both would
join the independent set. 
We will show that, if $b = o(\log n)$, the probability that no messages are successfully heard by $v$ or $w$ is $n^{-o(1)}$.
Since our graph contains $n/4$ disjoint pairs of this type, whose success or failure is independent, 
it follows that, with probability $1 - o(1)$, at least one
such edge will fail to have either endpoint receive a message.  


We define a strategy as a distribution over infinite sequences over the set $\{S, T, L\}$ with at most $b$ occurrences of $T$ and $L$ combined, where $S$, $T$, and $L$ correspond to sleep, transmit, and listen, respectively. A randomized algorithm (that uses at most $b$ energy) corresponds to a strategy (or distribution) according to which each node will sample a sequence and follow it until it hears a message or a collision. All nodes are running the same algorithm, thus following the same strategy.

Let $A_{u,x}$ be the event that a node $u$'s chosen random sequence (over $\{S, T, L\}$) agrees with an infinite sequence $x$ over the set $\{T, L\}$ in each of its $b$ occurrences of $T$ or $L$, and let $I_{u,x}$ be the corresponding indicator random variable.
Given any choice of string by $u$, if $x$ is sampled uniformly at random, then   $\Pr(A_{u,x}) \ge 2^{-b}$, since only the indices where $u$'s string has $T$ or $L$ are relevant. Therefore $\Exp_x(I_{u,x}) \ge 2^{-b}$. Since this did not depend on the random choices made by $u$, we have $\Exp_u(\Exp_x (I_{u,x})) \ge 2^{-b}$.
Switching the order of the summation gives $\Exp_x(\Exp_u(I_{u,x})) \ge 2^{-b}$. 
By the probabilistic method, there exists an infinite sequence $x^*$, for which $\Exp_u(I_{u,x^*}) \geq 2^{-b}$; \emph{i.e.}, $\Pr(A_{u,x^*})\geq 2^{-b} $

Since nodes $v$ and $w$ sample their sequences (over $\{S, T, L\}$) independently and from the same distribution,
$ \Pr(A_{v,x^*} \text{ and } A_{w,x^*}) \geq 4^{-b} $.

The intersection of the events $A_{v,x^*}$ and $A_{w,x^*}$ implies that neither $v$ nor $w$ heard a message. During any timestep at which both were awake, their sequences agreed with the shared sequence $x^*$, so they both transmitted or listened. In this case, $v$ and $w$ would need to join the output set by the comments above. The probability that such an event occurs for at least one edge is given by
\begin{align*}
    \Pr(\exists \text{ edge } (v,w): \text{ $v$ and $w$ both join})) & = 1 - \Pr(\forall \text{ edges } (v,w): \text{ $v$ and $w$ do not both join})) \\
    & \geq 1 - {\left( 1 - 4^{-b} \right)}^{n/4} 
    \geq 1 - e^{-n/4^{b+1}}.
\end{align*}
This gives the algorithm at least a failure probability of at least $1 - e^{-1/4}$ if $b \le 1/2 \log n$, so MIS requires $\Omega(\log n)$ energy.
\end{proof}

\section{Energy-Optimal MIS in the Collision Detection (CD) Model}
\label{sec:CDMIS}

In the CD model, the best known algorithm solves MIS in $O(\log^2 n)$ round complexity, whereas the best known round complexity lower bound is $\Omega(\log n)$. The gap between both round complexities remains open, but in this section, we settle the energy complexity of MIS. More precisely, we show that MIS can be solved (energy-optimally) in $O(\log n)$ energy.

\begin{restatable}{theorem}{MISCD}
\label{thm:MISCD}
In the CD model, Algorithm~\ref{alg:MIS} outputs an MIS with probability at least $1- 1/n$. Moreover, it does so using $O(\log n)$ energy and in $O(\log^2 n)$ rounds.
\end{restatable}

\begin{algorithm}
\caption{MIS Algorithm in the CD model.} \label{alg:MIS}
\begin{algorithmic}[1] 
\State $status \gets \undecided{}$
\For{Luby phase $i \gets 1$ to $C \log n$} \Comment{$C$ and $\beta$ control the success probability}
    \State{ $x \gets $ random string of $\beta \log n$ bits} 
\label{line:randomRank}
    \For{Bitty phase $j \gets 1$ to $\beta \log n$}
        \If{$x_j = 1$} 
            \State{ \transmit \ 1}
        \Else
            \State{ \listen }
            \If{heard 1 or collision}
                \State{\sleep \ for $\beta \log n - j$ rounds}
                \State{\textbf{break} and jump to line 16}
            \EndIf
        \EndIf
    \EndFor
    \State{}    \Comment{Next line reached only if normal $j$ loop termination}
    \State{\transmit \ 1} \Comment{Confirm inclusion in the MIS}  
    \State{$status \gets \inMIS{}$}
    \State \terminate
    \State{}  \Comment{Next line reached only if broke out of $j$ loop}
    \State{\listen} \Comment{Final check in the current round}
    \If{heard 1 or collision}
        \State $status \gets \notInMIS{}$
        \State \terminate
    \EndIf
\EndFor
\end{algorithmic}
\end{algorithm}

\subsection{Algorithm Description} 
\label{sec:cddesc}
Our energy-optimal MIS algorithm (see Algorithm~\ref{alg:MIS} and Figure~\ref{fig:CD-flow} for a graphic illustration) runs in $C \log n$ \emph{Luby phases} (Lines 3--20), each taking $\beta \log n + 1 $ rounds (including sleeping and awake rounds) --- where $C$ and $\beta$ are constants that control the success probability. Conceptually, we can separate each Luby phase into two parts: the \emph{competition} part, consisting of the first $\beta \log n$ rounds, and the \emph{checking} part, consisting of the last round. (The same conceptual structure applies to our later algorithms, see Section \ref{sec:noCDMIS}.) 

The competition separates nodes into \emph{winning} and \emph{losing} nodes, and winners are added to the MIS in this Luby phase. 
Conditioned on the high probability that the ranks are distinct, the competition ensures that winners form an independent set among the non-terminated nodes.
After all $\beta \log n$ rounds of the competition, the checking ensures that nodes in the computed independent set, and their neighbors, terminate respectively in and out of the MIS. More precisely, any node $u$ that wins the competition enters the MIS, sends a message to all neighbors in the phase's last round, and terminates. Meanwhile, any node $u$ that loses the competition listens in the last round and checks whether one of its neighbors won (i.e., if $u$ hears a 1 or a collision). If that is the case, node $u$ terminates as a node not in the MIS. Otherwise, it continues to the next Luby phase.

Next, we describe the competition. At the start, each non-terminated node $u$ is awake and generates a sequence of $\beta \log n$ random bits independently. Call this $\beta \log n$ bit binary number the rank of $u$, denoted by $x_u$. 
Then, during the competition, node $u$ determines whether its rank is smaller than the rank of any of its neighbors. This is done by a bit-by-bit comparison, using $\beta \log n$ \emph{Bitty phases}. In the first such phase, node $u$ examines the first bit of $x_u$, and if the bit is 1 then $u$ transmits 1, otherwise $u$ listens. In the latter case, if $u$ hears a 1 or a collision, then $u$ sleeps for all remaining Bitty phases (and has lost the competition). As for any subsequent phase $i$, any node $u$ that hasn't gone to sleep executes the same procedure, but considering the $i$th bit of $x_u$. Finally, any node that has not gone to sleep in the $\beta \log n$ Bitty phases has won the competition. 

Finally, we make two remarks. First, unlike the classical Luby's algorithm~\cite{luby1985simple}, in this version with the bit-by-bit competition, some winners may not be local maxima. However, the set of winners is a superset of the local maxima. Hence, the correctness of Luby's algorithm implies 
(see the proof of Lemma \ref{lem:edgeProgressCD}) 
that of this bit-by-bit version.
Second, since in our algorithm, only the act of transmission matters, and the actual messages play no role, the algorithm can also be implemented in the beeping model with the same round and energy complexities.
More concretely, in the pseudocode of Algorithm \ref{alg:MIS}, one can replace ``transmit 1'' with ``beep'' and ``heard 1 or a collision'' with ``heard a beep'' since in the beeping model, a listening node hears a beep if at least one of its neighbors beep.

\subsection{Analysis} 

We break our analysis of Algorithm~\ref{alg:MIS} into two main steps. The proof of Theorem~\ref{thm:MISCD} appears in Appendix~\ref{app:proofs}.  First, we prove that, with high probability, it finds an independent set.

\begin{lemma} \label{lem:CD-output-independent}
Let $S = \{v \colon status(v) = \inMIS{}\}$ be the set output by Algorithm~\ref{alg:MIS}. Then
$
\Prob{S \mbox{ is an independent set}} \ge 1 - \frac{1}{2n}.
$
\end{lemma}

Next, we show that, with high probability, every node is within one hop of at least one vertex output by our algorithm. To do so, we introduce the concept of a residual graph and follow part of a standard proof of Luby's algorithm, which says that, conditioned on whatever happened previously, one phase of Luby's algorithm shrinks the size of the residual graph by at least half, in expectation.

\begin{definition}
    Let $V_0 = V$ and $1 \le i \le C \log n$. Let $V_i$ denote the set of vertices that at the end of Luby phase $i$, have not yet terminated, or equivalently have $status = \undecided{}$. We call the subgraph of $G$ induced by $V_i$ the \emph{residual graph} at the end of Luby phase $i$, and denote it $G_i = (V_i, E_i)$.
\end{definition}

\begin{lemma}
\label{lem:edgeProgressCD}
    Let $1 \le i \le C \log n$. Then, $ \ExpCond{|E_i|}{E_{i-1}} \le \frac{|E_{i-1}|}{2}$.
\end{lemma}

We remark that the proof of the above lemma (see Appendix~\ref{app:proofs}) does not imply that the running time of Luby's algorithm stochastically dominates the running time of Algorithm~\ref{alg:MIS}.  This is because the monotonicity we exploit in our proof only exists within a single Luby phase; not across multiple Luby phases.  However, we do get the following immediate corollary.

\begin{corollary} \label{cor:Luby-done}
   $ \Prob{E_{C \log n} = \emptyset} \ge 1 - \frac{1}{2n}.$
\end{corollary}

\begin{lemma}
\label{lem:CD-output-maximal}
Let $S = \{v \colon status(v) = \inMIS{}\}$ be the set output by Algorithm~\ref{alg:MIS}.  Then
$
\Prob{S \cup N(S) = V} \ge 1 - \frac{1}{2n}.
$    
\end{lemma}



\section{Auxiliary Primitives in the no-CD Model}
\label{subsec:noCDPrimitives}

We give several primitives in the no-CD model, and these will be key components of our no-CD MIS algorithm in Section \ref{sec:noCDMIS}. More concretely, we first give several energy-efficient backoff procedures. Then, we give an improved runtime version of the algorithm from~\cite{davies2023}.

\subsection{Energy-Efficient Backoff Procedures}
\label{subsec:backoff}

We design slightly more energy-efficient sender- and receiver-backoff procedures, which are used to adapt our CD algorithm to work in the no-CD model. These are described in more detail in Appendix~\ref{app:backoff} and have the following properties, which are proven in Appendix~\ref{app:backoff}.

\begin{lemma}
\label{lem:ETRBackoffUpperBounds}
    Let $k$ be any positive integer. Both (sender-side and receiver-side) energy-efficient $k$-repeated backoffs take $O(k \log \Delta)$ rounds. Moreover, any node $v$ calling 
    \begin{itemize}
        \item $\Call{\SndBackoff{}}{k,\Delta}$ is awake for $k$ rounds.
        \item $\Call{\RecBackoff{}}{k,\Delta, \DeltaEst}$ is awake for $O(k \log \DeltaEst)$ rounds.
    \end{itemize}
\end{lemma}

\begin{lemma}
\label{lem:ETRBackoffCorrectness}
    For any node $v$ calling $\Call{\RecBackoff{}}{k,\Delta, \DeltaEst}$ with at most $\DeltaEst$ neighbors calling $\Call{\SndBackoff{}}{k,\Delta}$ simultaneously, it holds with probability at least $1-(7/8)^k$ that node $v$ returns true if and only if at least one of its neighbors simultaneously called $\Call{\SndBackoff{}}{k,\Delta}$.     
\end{lemma}

\subsection{Round-Efficient MIS in the no-CD Model}
\label{subsec:DaviesAlg}

As stated in the related work, Davies~\cite{davies2023} gives an $O(\log^3 n)$ round algorithm for MIS in the no-CD model. We make several (minor) modifications 
to improve its runtime to $O(\log^2 n \log \Delta)$ rounds, where $\Delta$ is an upper bound on the maximum degree of the graph.
We call the improved algorithm \peterAlg{} and give a short description below. It serves as one of the components of our $O(\log^2 n \log \log n)$ energy MIS algorithm given in Section~\ref{sec:noCDMIS}.

Before we describe the minor improvements, note that in~\cite{davies2023}, each phase simulating one round of Ghaffari's MIS algorithm is also called a ``round,'' whereas timesteps are what we call rounds in our model.
First, the number of timesteps in the $\Call{Decay}{}$ subroutine can be reduced to $\Theta(\log \Delta)$. Even then, $O(\log n)$ iterations of this shorter $\Call{Decay}{}$ subroutine results in the same high probability success guarantee.
Second, the $\Call{EstimateEffectiveDegree}{}$ subroutine can be run for only $\Theta(\log \Delta)$ outer loop iterations (each still consisting of $O(\log n)$ timesteps) while maintaining the same high probability success guarantee. The rest of the analysis remains unchanged.

\section{Energy-Efficient MIS in the no-CD Model}
\label{sec:noCDMIS}

In the no-CD model, the best known round and energy upper bound is $O(\log^3 n)$ \cite{davies2023}, or if we parametrize by $\Delta$, $O(\log^2 n \log \Delta)$ (see Section \ref{subsec:DaviesAlg}).
On the other hand, the only known energy lower bound in this setting is the $\Omega(\log n)$ lower bound from the CD model. Hence, there is a $\Theta(\log n \log \Delta)$ gap between the known upper and lower bounds in the no-CD model. In this section, we present an MIS algorithm for no-CD (Algorithm \ref{alg:no-CD-MIS}) with a significantly better energy complexity of $O(\log^2 n \log \log n)$, and as a result we narrow this gap by a logarithmic factor, down to $\Theta(\log n \log \log n)$.

\begin{restatable}{theorem}{MISnoCD}
\label{thm:MISnoCD}
In the no-CD model, Algorithm \ref{alg:no-CD-MIS} outputs an MIS with probability at least $1-1/n$.
Moreover, it does so using $O(\log^2(n) \log \log n)$ energy and in $O(\log^3 n \log \Delta)$ rounds.
\end{restatable}

\subsection{Insights into our Algorithm}

To solve MIS in the no-CD setting, one option is to take an MIS algorithm in the CD model (say, Algorithm~\ref{alg:MIS}) and simulate it using traditional backoff (described in Section~\ref{subsec:backoff}). However, we would need to simulate each round with high probability, leading to an $O(\log n \log \Delta)$ blow-up in the round complexity but also, and most importantly, in the energy complexity.
We identify two areas that drive the energy cost up by a factor of $O(\log n \log \Delta)$ when we take the above approach with Algorithm~\ref{alg:MIS} and make non-trivial adaptations to address them. 

\subsubsection{Competition}
The first problematic place is in the inner loop of the competition, where nodes with a `0' bit listen to determine if they should drop out. An eventual winner is one that survives the Luby phase without hearing any of its neighbors. Note that just as any other node, it is likely to have $\Theta(\log n)$ 0's in its bitstring (since each bit is chosen uniformly at random and independently). But, because an eventual winner never hears a neighbor, it will listen for all rounds in all backoffs corresponding to `0' bits. As a result, during this Luby phase, it will have spent $O(\log^2 n \log \Delta)$ energy.

A similar issue may affect some eventual losers, as well.
For example, two adjacent nodes may happen to choose random strings that agree in their first $\Theta(\log n)$ bits. If they have no other neighbors to knock them out of competition, both nodes will spend $\Theta(\log^2 n \log \Delta)$ energy listening to silence during their $0$ bit rounds, before one eventually loses.

To fix the above issues, we give each node an energy budget of $O(\log n \log \Delta)$. If a node is awake for the entire backoff corresponding to its first `0' bit, then it has used up a significant portion of its budget for the entire algorithm. It cannot afford to use this much energy repeatedly, so it will do two things to reduce its energy use: it will play the remainder of the Luby phase with the (justified) assumption that it has only $O(\log n)$ surviving neighbors, so it can shorten how long it listens in the backoffs. Additionally, even if it is knocked out of the competition, it will commit itself to terminating in this Luby phase. The nodes that are thus committed induce a subgraph of maximum degree $O(\log n)$, so at the end of the Luby phase, they can afford to run an algorithm that is energy-efficient on small degree graphs (e.g,. the naive simulation of Algorithm~\ref{alg:MIS}, or for better runtime, the algorithm from Section~\ref{subsec:DaviesAlg}).
Each committed node with no neighbor in the output set runs this subroutine exactly once, bringing the energy complexity to $O(\log^2 n \log \log n)$.

\begin{figure}[H]
\centering
\definecolor{flowchartorange}{RGB}{245 166 35}
\definecolor{flowchartred}{RGB}{218, 41, 28}
\definecolor{flowchartblue}{RGB}{0, 156, 189}
\definecolor{flowchartgreen}{RGB}{132, 189, 0}
\definecolor{flowchartpurple}{RGB}{125, 85, 199}
\begin{minipage}[c]{0.6\linewidth}
\centering
\includegraphics[scale=0.4]{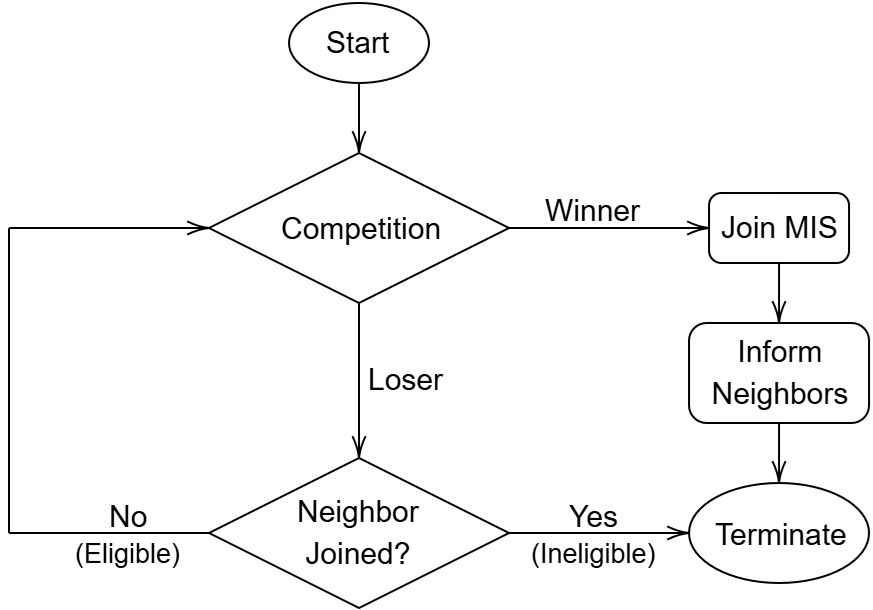}
    \caption{Flowchart for our CD algorithm, Algorithm~\ref{alg:MIS}.}
        \label{fig:CD-flow}
\end{minipage}

\vspace{0.5cm}
\begin{minipage}[c]{\linewidth}
\centering
\includegraphics[width=0.8\linewidth]{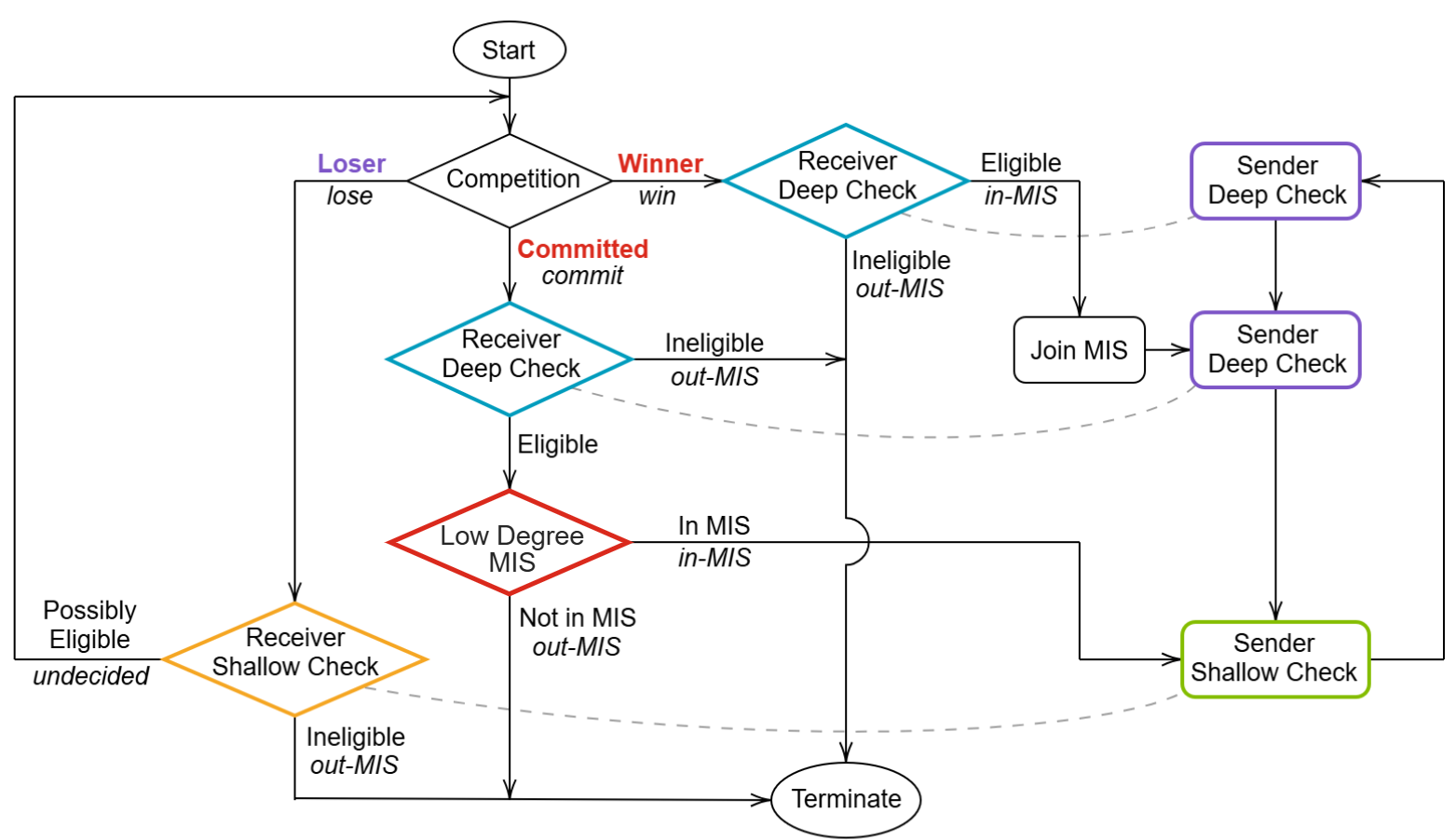}
    \caption{Flowchart for our no-CD algorithm, Algorithm~\ref{alg:no-CD-MIS}. Dashed lines indicate concurrent execution. Energy usage is color-coded as follows: {\color{flowchartred} $O(\log^2 n \log \log n)$}, {\color{flowchartblue} $O(\log n \log \Delta)$}, {\color{flowchartpurple} $O(\log n)$}, {\color{flowchartorange} $O(\log \Delta)$}, {\color{flowchartgreen} $O(1)$}.}
    \label{fig:noCD-flow}
\end{minipage}%
\end{figure}



\subsubsection{Checking}
The second area of concern is in the notification process at the end of each Luby phase.
All nodes who lost in the competition now listen, because they might have a neighbor in the output set. However, if they have no such neighbors, then they listen for the full $O(\log n \log \Delta)$ rounds. Crucially, if this happens in every phase, nodes would end up using too much energy. And yet, this seems necessary, as we should ensure that any node that neighbors an MIS node no longer participates in the following competitions.

However, we give up on that property. More concretely, at the end of each Luby phase, losing nodes perform a \emph{``shallow'' check} to detect the presence of MIS neighbors with only constant probability --- via a single iteration of backoffs. This gives the neighbors of MIS nodes a constant probability to drop out, but at a vastly reduced cost. On the other hand, a winning node performs a thorough, \emph{``deep'' check} --- via $O(\log n)$ iterations of backoffs --- to detect the presence of MIS neighbors with high probability --- to decide whether it joins the output set (when there are no such neighbors) or sets its status to ``not in MIS.''

\subsection{Algorithm}



Next, we give our MIS algorithm: Algorithm \ref{alg:no-CD-MIS}. Figure~\ref{fig:noCD-flow} provides a graphic illustration of the flow control for Algorithm~\ref{alg:no-CD-MIS}. This algorithm relies on several constants. Some ensure our algorithm successfully computes an MIS with probability at least $1-\sfrac{1}{n}$. Increasing these constants yields better success probabilities. In particular, we choose $\beta \geq 4$, $\kappa \geq 5$ and $C \geq 4/\log(\sfrac{64}{63})$, whereas we choose $C'$ such that $\Call{\RecBackoff{}}{C' \log n,\Delta}$ succeeds with probability $1-\sfrac{1}{n^5}$. As for $\peterAlg{}$, we ensure it succeeds with probability $1-\sfrac{1}{n^2}$. 

Our other constants ensure that nodes stay synchronized throughout the Luby phase, and upper bound the round complexity for
\begin{itemize}
    \item $\Call{\SndBackoff{}}{k,\Delta}$ and  $\Call{\RecBackoff{}}{k,\Delta}$: $T_B(k) = k \lceil \log \Delta \rceil$,
    \item $\Call{Competition}{\Delta}$: $T_C = \beta \log^2 n \log \Delta$,
    \item $\peterAlg{}$ on an induced subgraph of max degree $\kappa \log n$: $T_G = O(\log^2 n \log \log n)$,
    \item A single Luby phase: $T_L = T_C + 2 T_B(C' \log n) + T_G + T_B(1) = O(\log^2 n \log \Delta)$.
\end{itemize}

\begin{algorithm}
\caption{Distributed Maximal Independent Set (MIS) no-CD Algorithm}\label{alg:no-CD-MIS}
\begin{algorithmic}[1]
\State $status \gets \undecided{}$
\For{Luby phase $i \gets 1$ to $C \log n$}
    \If{$status = \undecided{}$}
         \Call{Competition}{$\Delta$}
    \Else
         \; \sleep{} until round $(i - 1) T_L + T_C$ 
    \EndIf

    \Statex 

    \State{} $heard \gets$ False   \Comment{All nodes are synchronized here} 
    
    \If{$status = \inMIS{}$}
         $\Call{\SndBackoff{}}{C' \log n,\Delta}$
    \ElsIf{$status = \win{}$}
        \State $heard \gets \Call{\RecBackoff{}}{C' \log n,\Delta}$ \Comment{Deep check for MIS neighbors}
        \If{heard}
            $status \gets \notInMIS{}$, then terminate early
        \Else
            \; $status \gets \inMIS{}$ \label{line:decide1}
        \EndIf
    \Else
         \; \sleep{} until round $(i - 1) T_L + T_C + T_B(C' \log n) $
    \EndIf
 
    \Statex 
    
    \If{$status = \inMIS{}$} \Comment{All nodes are synchronized here}
        \State $\Call{\SndBackoff{}}{C' \log n,\Delta}$
        \State{\sleep{} until round $(i - 1) T_L + T_C + 2 T_B(C' \log n) + T_G$}
    \ElsIf{$status = \commit{}$} 
        \State $heard \gets \Call{\RecBackoff{}}{C' \log n,\Delta}$ \Comment{Deep check for MIS neighbors}
        \If{heard}
            \State $status \gets \notInMIS{}$, then terminate early
        \Else
            \State $status \gets \undecided{}$
            \State \peterAlg{} \Comment{Run on subgraph of maximum degree $O(\log n)$} \label{line:decide2}
        \EndIf
    \Else
         \; \sleep{} until round $(i - 1) T_L + T_C + 2 T_B(C' \log n) + T_G$
    \EndIf
    
    \Statex 

    \If{$status = \inMIS{}$} \label{line:shallowCheck}
        $\Call{\SndBackoff{}}{1,\Delta}$ \Comment{All nodes are synchronized here}
    \Else
        \State $heard \gets \Call{\RecBackoff{}}{1,\Delta}$ \Comment{Shallow check for MIS neighbors}
        \If{heard}
            $status \gets \notInMIS{}$, then terminate early
        \Else
            \; $status \gets \undecided{}$
        \EndIf
    \EndIf
\EndFor
\end{algorithmic}
\end{algorithm}

\begin{algorithm}
\caption{COMPETITION}
\begin{algorithmic}[1]
\Procedure{Competition}{$\Delta$}

\State{ $\DeltaEst \gets \Delta$,  $heard \gets$ False}
\State{ $x \gets $ random string of $\beta \log n$ bits}
\For{Bitty phase $j \gets 1$ to $\beta \log n$}
    \If{$status = \lose$} \sleep{}
    \ElsIf{$x_j = 1$}
        $\Call{\SndBackoff{}}{C' \log n,\Delta}$
    \Else
        \State $heard \gets heard \, \vee \, \Call{\RecBackoff{}}{C' \log n,\Delta, \DeltaEst}$ \Comment{Logical OR}
        \If{$heard$ and $status \ne commit$} 
            \State {$status \gets$ \lose}
        \ElsIf{not $heard$} \Comment{Not hearing implies $O(\log n)$ undecided neighbors.}
            \State{$\DeltaEst \gets \min\{\Delta, \kappa \log n\}$}  
            \State{$status \gets \commit$}
        \EndIf
    \EndIf
\EndFor
\If{not $heard$} \Comment{Nodes that heard nothing win, including committed ones.}
    \State{$status \gets \win$}
\EndIf
\EndProcedure
\end{algorithmic}
\end{algorithm}

\subsection{Properties of a Luby Phase}

We now show properties that pertain to a single Luby phase $i$. More precisely, we show properties on sets $C_i$ and $W_i$ --- these are defined as the set of undecided nodes that run the competition in the $i$th Luby phase and subsequently set their status to \commit{} and \win{}, respectively. First, we consider the sets $C_i$ and show that neighboring, committed nodes must have committed in the same Bitty phase with high probability. This helps us show that for any Luby phase $i$, $C_i$ induces a logarithmic degree subgraph. That is, with high probability, among the neighbors of any given committed node $v$, there can be at most $O(\log n)$ nodes with $status \ne \lose{}$. This justifies our reduction of the degree estimate to $\kappa \log n$ (in our case, $\kappa \geq 5$) whenever a node sets its status to \commit{}.

\begin{lemma} \label{lem:Ci-neighborsSimultaneousCommits}
Consider a single call to $\Call{Competition}{}$. Let $u,v$ be two neighboring nodes that set their status to \commit{}. Then, with probability at least $1 - 2/n^{5}$, $u$ and $v$ set their status to \commit{} in the same Bitty phase.

\end{lemma}

\begin{lemma} \label{lem:Ci-LowDegree}
Consider a single call to $\Call{Competition}{}$.
Let $\kappa$ be any strictly positive integer, and let $B$ be the event that there exists a node $v$ such that more than $\kappa \log n$ neighbors of $v$ do not have status \emph{\lose{}} in the (Bitty) phase of $v$'s first $0$ bit, and $v$ sets its status to \emph{\commit{}}. Then, $\Prob{B} \leq 4/n^{4}$.
\end{lemma}

\begin{corollary}
\label{cor:goodDegreeEstimate}
    For any Luby phase $i$, each statement below holds with probability $1-4/n^4$:
    \begin{enumerate}
        \item During the competition, for Bitty phase $j$ and node $v$ (such that $v$'s status is not \lose{}), the degree estimate of $v$ upper bounds the number of $v$'s awake neighbors (i.e., starting the Bitty phase with $status \neq \lose{}$).
        \item The subgraph induced by $C_i$ has maximum degree $O(\log n)$.
    \end{enumerate}
\end{corollary}

Let $C_i^*$ denote those nodes with status \commit{} but having not detected any MIS neighbor during Luby phase $i$. Then, Corollary~\ref{cor:goodDegreeEstimate} implies that we can run the $O(\log^2 n \log \Delta)$ round MIS algorithm described in Section \ref{subsec:DaviesAlg}, to compute an MIS on the subgraph (of the communication graph) induced by $C_i^*$ in $O(\log^2 n \log \log n)$ rounds (and energy).

We follow up by showing properties on the set $W_i$ for any Luby phase $i$. We show that in the competition, the undecided nodes with locally maximum bitstrings set their status to \win{} (and thus join set $W_i$) with high probability. This implies, among other things, that $W_i$ is not empty until no undecided nodes remain. Next, we show that in the competition, with high probability, no two neighbors set their status to \win{} (i.e., $W_i$ is independent).
 
\begin{lemma} \label{lem:Wi-notEmpty}
Consider a single call to $\Call{Competition}{}$. Let $v$ be an undecided node whose bitstring $x(v)$ is a local maximum. Then, 
$\Pr\left( \text{$v$ sets status to \win{}}\right) \ge 1 - 1/n^2$.
\end{lemma}

\begin{lemma} \label{lem:Wi-Independent}
For any Luby phase $i$ and two neighbors $u,v$: $\Pr\left( \text{$u \in W_i$ and $v \in W_i$} \right) \le 6/n^4$.
\end{lemma}

Finally, we highlight that any node that attempts to join the MIS following the competition (i.e., that is in $W_i \cup C_i$) is decided by the end of that Luby phase.

\begin{lemma} \label{lem:WiCi-MustDecide}
For any Luby phase $i$, any node in $W_i \cup C_i$ decides by the end of that phase with probability 1.
\end{lemma}

\subsection{Analysis}

We now prove our main result, Theorem \ref{thm:MISnoCD}. We start with auxiliary lemmas that help prove the correctness. Their proofs can be found in Appendix~\ref{app:proofs}. First, we show that the set of nodes with status \inMIS{} stays independent throughout the algorithm with high probability.

\begin{lemma} \label{lem:maintainsMIS}
For any Luby phase $i$, at the start, the nodes with status \inMIS{} form an independent set with probability at least $1-O(\log n)/n^2$.
\end{lemma}


It remains to show that the logarithmic number of Luby phases of Algorithm \ref{alg:no-CD-MIS} suffices for all nodes to become decided with high probability (as either \inMIS{} or \notInMIS{}).
To do so, we follow the lines of the classical Luby analysis. In other words, we consider the residual graphs, whose definition follows, and show that the number of edges in the residual graphs decreases by a constant fraction every phase (see Lemma \ref{lem:eliminate-edges}). 

\begin{definition}
    Let $V_0 = V$ and $1 \le i \le C \log n$. Let $V_i$ denote the set of vertices that at the end of Luby phase $i$, have $status \ne \notInMIS{}$. We call the subgraph of $G$ induced by $V_i$ the \emph{residual graph} at the end of Luby phase $i$, and denote it $G_i = (V_i, E_i)$.
\end{definition}

Note that, in contrast to the definition in Section \ref{sec:CDMIS}, the residual graph contains the decided MIS nodes --- as here, MIS nodes do not terminate early --- as well as undecided nodes (i.e., with $status = \undecided{}$) that have an MIS neighbor but do not know it yet. (This can happen because MIS nodes inform their neighbors via shallow checks, which only succeed with constant probability per phase.) In particular, the latter nodes complicate the analysis: they continue to participate in the competition but cannot enter the MIS, and yet, the following lemma shows that such nodes have a limited impact on the progress of the algorithm. In short, they lead to a constant factor slowdown only.

\begin{lemma}
\label{lem:eliminate-edges-ideal}
For any Luby phase $1 \le i \le C \log n$, $\ExpCond{X_i}{E_{i-1}} \ge \frac{|E_{i-1}|}{8}.$
\end{lemma}

After showing the above lemma, we can prove that every phase leads to, in expectation, a constant factor loss in the (edge) size of the residual graphs.

\begin{lemma} \label{lem:eliminate-edges}
For any Luby phase $1 \leq i \leq C \log n$,
$
\ExpCond{|E_{i}|}{E_{i-1}} \le \frac{63}{64}|E_{i-1}|.
$
\end{lemma}

Finally, we prove our main result.

\MISnoCD*

\begin{proof}
First, we show correctness. By Lemma~\ref{lem:maintainsMIS}, with probability at least $1-O(\log n)/n^2 \geq 1-\sfrac{1}{4n}$, the set of nodes with status \inMIS{} is independent throughout the execution of Algorithm~\ref{alg:no-CD-MIS}. Hence, it suffices to show that all nodes become decided within $C \log n$ Luby phases. 
By induction and Lemma~\ref{lem:eliminate-edges}, for every $i \ge 1$, we have 
$\Exp(|E_i|) \le \left(\frac{63}{64}\right)^i |E_0|$. 
Hence, by Markov's inequality,
$\Prob{|E_i| \ge 1} \le \Exp(|E_i|)  \le \left(\frac{63}{64}\right)^i |E_0|$. By choosing $C \geq 4/\log(64/63)$, we get that $\Prob{E_{C \log n} = \emptyset} \ge 1 - \sfrac{1}{4n}$. Finally, we consider any Luby phase $i$ with $E_i = \emptyset$. In that phase's competition, any undecided node chooses a local maximum bitstring and enters $W_i$ with probability at least $1 - \sfrac{1}{4n}$ by Lemma \ref{lem:Wi-notEmpty}, in which case it becomes decided by the end of the phase with probability 1, by Lemma \ref{lem:WiCi-MustDecide}. (A final $\sfrac{1}{4n}$ probability term comes from thresholding the energy complexity, as explained at the end of the proof.) In summary, all nodes become decided, and the output is an MIS with probability at least $1-\sfrac{1}{n}$.

Second, we bound the round complexity. From the algorithm description, each Luby phase takes $T_L = O(\log^2(n) \log \Delta)$ rounds. 
Hence, the round complexity of the algorithm (which runs for $C \log n$ Luby phases) is $O(\log^3 n \log \Delta)$ rounds.  

Finally, we upper bound the energy complexity. First, if any node $v$ starts the Luby phase as an MIS node then $v$ spends $O(\log n)$ energy: $v$ sleeps during the competition, and its participation in sender backoffs during the checking adds up to $O(\log n)$ energy by Lemma \ref{lem:ETRBackoffUpperBounds}. 

Second, if node $v$ starts \undecided{} and enters $C_i \cup W_i$, then during that Luby phase, $v$ spends $O(\log^2(n) \log \log n)$ energy. Indeed, during the competition, node $v$ uses $O(\log^2 n)$ energy for the sender backoffs overall (by Lemma \ref{lem:ETRBackoffUpperBounds}), and $O(\log^2 n) + O(\beta \log n \cdot \log(n)\log\log n)$ energy for the receiver backoffs overall (by Lemma \ref{lem:ETRBackoffUpperBounds} and due to the change in the degree estimate). As for the checking, node $v$ uses up $O(\log^2(n) \log \log n)$ energy during \peterAlg{} due to 
Davies' algorithm (cf. Section \ref{subsec:DaviesAlg}) 
and
because the subgraph induced by $C_i$ has maximum degree $O(\log n)$ (see Corollary \ref{cor:goodDegreeEstimate}).

By Lemma \ref{lem:WiCi-MustDecide}, there can be a single Luby phase in which $v$ starts \undecided{} and enters $C_i \cup W_i$, as subsequently node $v$ either sets its status to \inMIS{}, or sets its status to \notInMIS{} and sleeps for the remainder of the MIS algorithm. Hence, summing over all Luby phases in which $v$ either starts as an MIS node, or attempts to join the MIS (i.e., joins $C_i \cup W_i)$, the energy spent by $v$ is upper bounded by $O(\log^2(n) \log \log n)$.

It remains to bound the energy complexity spent over any Luby phases in which $v$ starts undecided and loses the competition. Let $A_i$ be the energy spent by $v$ in Luby phase $i$ times the indicator random variable that $v$ loses Luby phase $i$. We can upper bound the energy complexity $A = \sum_{i=1}^{\beta \log n} A_i$ by $O(\log^2 n)$ with high probability. Note that the energy spent in Luby phases in which $v$ loses is either spent sending during the leading 1 bits or listening during the first 0 bit. Let $X_i$ be the number of leading 1 bits in $v$'s bitstring for Luby phase $i$. The random variables $(X_i + 1)$ are independent geometric random variables with parameter $\frac{1}{2}$. In the Bitty phase corresponding to $v$'s first 0 bit in Luby phase $i$, let $N_i$ be the number of $v$'s neighbors that are sending because they are still active and have a 1 bit and let $B_i$ be the number of backoffs $v$ must participate in until hearing a message. Let $I_i$ be the indicator random variable that $N_i > 0$ and let $Y_i = B_i I_i  + (1 - I_i)$. Conditioned on $N_i$ and any random variables $Y_j$ for $j < i$, $Y_i$ is a geometric random variable with parameter $p(N_i)$, where $p(0) = 1$ and $p(n) \ge \frac{1}{8}$ for $n \ge 1$, which follows from well-known statements on exponential backoff (see Lemma~\ref{lem:ETRBackoffCorrectness} in Section \ref{subsec:backoff}).
Then, $A_i \le O(\log n) X_i + O(\log \Delta) Y_i$.

For the $X_i$ random variables, we instead bound the sum of $(X_i+1)$ random variables, since they are more nicely defined. Let $X = \sum_{i=1}^{\beta \log n} (X_i + 1)$. Then, $\mu_X = \Exp[X] = 2\beta \log n$. By Theorem~\ref{Chernoff Bound}, we have $\Prob{X \ge 2 \beta \lambda \log n} \le e^{-\beta \log n (\lambda - 1 - \ln \lambda)} \le n^{-C_X}$ for arbitrary $C_X$ and sufficiently large $\lambda$. Next, we want to bound the sum of the $Y_i$ random variables. However, these $Y_i$ random variables are not independent since the number of backoffs required to hear a message depends on the number of neighbors sending, and the number of neighbors sending in any Luby phase affects how many neighbors send in the next Luby phase.
Hence, we instead show by induction that $Y = \sum_{i=1}^{\beta \log n} Y_i$ is stochastically dominated by $Z = \sum_{i=1}^{\beta \log n} Z_i$, where the random variables $Z_i$ are independent geometric random variables with parameter $\frac{1}{8}$. For the base case, it can be easily seen that $Y_1 \preceq Z_1$. Suppose $\sum_{i=1}^{k-1} Y_i \preceq \sum_{i=1}^{k-1} Z_i$. Then,
 \allowdisplaybreaks
 \begin{align*}
    \Prob{\sum_{i=1}^k Y_i \ge y} & = \sum_{y', n} \Prob{\sum_{i=1}^k Y_i \ge y \, \middle\vert \, N_i = n, \sum_{i=1}^{k-1} Y_i = y'} \Prob{N_i = n, \sum_{i=1}^{k-1} Y_i = y'} \\
    & = \sum_{y', n} \Prob{Y_k \ge y - y' \, \middle\vert \, N_i = n, \sum_{i=1}^{k-1} Y_i = y'} \Prob{N_i = n, \sum_{i=1}^{k-1} Y_i = y'} \\
    & = \sum_{y', n} \min \left\{ 1, {(1 - p(n))}^{y - y' - 1} \right\} \Prob{N_i = n, \sum_{i=1}^{k-1} Y_i = y'} \\
    & \le \sum_{y', n} \min \left\{ 1, {\left( \frac{7}{8} \right)}^{y - y' - 1} \right\} \Prob{N_i = n, \sum_{i=1}^{k-1} Y_i = y'} \\
    & = \Prob{Z_k + \sum_{i=1}^{k-1} Y_i \ge y} 
 \le \Prob{\sum_{i=1}^k Z_i \ge y},
 \end{align*}
where the first inequality comes from $p(n) \ge \frac{1}{8}$ and the second comes from Lemma~\ref{lem:stochasticDominationSum}. Applying Theorem~\ref{Chernoff Bound} to $Z$ with $\mu_Z = \Exp[Z] = 8 \beta \log n$, we get $\Prob{Z \ge 8 \beta \lambda \log n} \le e^{-\beta \log n (\lambda - 1 - \ln \lambda)} \le n^{-C_Z}$ for arbitrary $C_Z$ and sufficiently large $\lambda$. Following which, the stochastic domination implies the same bound holds for $\Prob{Y \ge 8 \beta \lambda \log n}$. 

Finally, a union bound (over the nodes) shows that $A = O(\log^2 n)$ with high probability, say $1-\sfrac{1}{4n}$. Adding up the energy complexities, we get that all nodes spend $O(\log^2(n) \log \log n)$ energy with probability at least $1-\sfrac{1}{4n}$ during this MIS algorithm. To obtain the claimed deterministic upper bound on the energy complexity, we can do the following: any node spending more than some $\Theta(\log^2(n) \log \log n)$ energy threshold simply sleeps for the remainder of the algorithm and arbitrarily decides whether to join the MIS.
\end{proof}

\section{Conclusion and Open Questions}
\label{sec:conc}

We presented new, more energy-efficient MIS algorithms for radio networks.
While our CD algorithm is energy-optimal, it is unclear whether our no-CD algorithm is. 
A key open problem is whether we can improve the energy complexity significantly in the no-CD model or whether our bound of $O(\log^2 n\log \log n)$ is nearly optimal (up to a $O(\log \log n)$ factor). 

In the CD model, it is known that one can design an $O(\log n)$ round algorithm for growth-bounded graphs \cite{wattenhofer-growthbounded}. 
It is not clear if this bound can be achieved for general graphs as well. The round complexity of our CD algorithm is $O(\log^2 n)$, which is off by an $O(\log n)$ factor of the $\Omega(\log n)$ round complexity lower bound in the CD model \cite{wattenhofer-growthbounded}. Can we improve the round complexity in the CD model while getting the optimal $O(\log n)$ energy bound?

In the no-CD model, our energy complexity bound of $O(\log^2 n \log \log n)$ almost matches the round complexity lower bound of $\Omega(\log^2 n)$. Can we design a better energy-efficient algorithm that takes $o(\log^2 n)$ energy in arbitrary graphs? 
Furthermore, can we improve the round complexity of our no-CD algorithm while maintaining its energy complexity?




\appendix

\section{Some Useful Facts from Probability Theory}

We will need the following concepts and results from Probability Theory.
\begin{definition}
    \textbf{Stochastic Domination.} For any two random variables $X, Y$, we say that $X$ stochastically dominates $Y$, denoted by $X \succeq Y$, if for all $z \in \mathbb{R}$, $\Pr[X > z] \geq \Pr[Y > z]$.
\end{definition}

In general it is not true that stochastic domination carries over to sums of random variables. However, the following lemma shows that it is true in a limited setting. We leave its proof as an exercise to the reader.

\begin{lemma}
\label{lem:stochasticDominationSum}
Let $X$ and $Y$ be random variables such that $X \succeq Y$ and let $Z$ be a random variable that is independent of $X$ and $Y$. Then, $X + Z \succeq Y + Z$.
\end{lemma}

We will also need the following concentration inequality for the sum of independent geometric random variables. 

\begin{theorem} [Theorem 2.1 from~\cite{janson}]\label{Chernoff Bound}
    Let $X_1, X_2, \ldots, X_n$ be independent geometric random variables with parameters $p _1, p_2, \ldots, p_n$ respectively. Let $X = \sum_{i=1}^n X_i$, $\mu = E[X]$, and  $p_* = \min\limits_{i} p_i$. Then, for any $\lambda \geq 1$ the following holds: 
    $$
        \Pr\left( X \geq \lambda \mu \right)  \leq e^{-p_* \mu (\lambda - 1  - \ln \lambda)}
    $$
\end{theorem}

\section{Proofs}
\label{app:proofs}



    


\begin{proof} [Proof of Lemma~\ref{lem:edgeProgressCD}]
First, observe that, in each Luby phase, Algorithm~\ref{alg:MIS} always adds all local maxima of the random function $x$ to the output.  This is because, for a local maximum $v$, every neighbor has a zero in the first Bitty phase when it disagrees with $v$.  Consequently, if you started both Algorithm~\ref{alg:MIS} and Luby's algorithm from the same residual graph
$G_{i-1}$, Algorithm~\ref{alg:MIS} would remove a superset of the vertices removed by Luby's algorithm. Since we are looking at vertex-induced subgraphs, this implies we also remove a superset of the edges removed by Luby's algorithm.  Since Luby's algorithm already satisfies the conclusion of the Lemma, so does Algorithm~\ref{alg:MIS}.
\end{proof}





\begin{proof} [Proof of Theorem~\ref{thm:MISCD}] 
The correctness and round complexity claims are straightforward. For the correctness claim, by Lemma~\ref{lem:CD-output-independent}, with high probability, the output is an independent set. When it is, Lemma~\ref{lem:CD-output-maximal} tells us that, with high probability, it is also maximal. Thus, a union bound over the failure probabilities leads to the claimed correctness guarantee. As for the round complexity, Algorithm~\ref{alg:MIS} consists of two nested for-loops, each with $O(\log n)$ iterations, which implies the claimed $O(\log^2 n)$ bound.

Finally, we show the energy complexity upper bound. More concretely, we show that, with probability $1 - 1/\poly(n)$, no node spends more than $O(\log n)$ energy. Since our algorithm already has a $1/\poly(n)$ chance of failure, we can think of exceeding the energy bound as another way to fail, in which case our bound on the maximum energy spent is absolute.  

Let us examine the energy spent by a particular node, $v$.  We split our analysis into two parts: early rounds, in which $v$'s decision is still in doubt, and late rounds, in which $v$'s decision is determined, although $v$ may not know it yet.  Specifically, a \emph{late round} is a Bitty phase for which $v$ is still active, but all neighbors of $v$ have already dropped out.  In this case, $v$ will inevitably complete the inner for loop, and set $status = \inMIS{}$.  An \emph{early round} is a Bitty phase in which $v$ is active, and at least one neighbor of $v$ is also active.
Note that every round in which $v$ spends energy within a Luby phase is either early or late, but not both.

Since all the late rounds must occur within a single Luby phase, these must contribute at most $\beta \log n$ to $v$'s energy expenditure.

We now examine the early rounds one by one, always conditioning on the outcomes of all previous rounds.
Say that an early round is \emph{fruitful} if $v$'s bit of $x$ for that round is a $0$, and at least one active neighbor's bit for that round is a $1$.  Since there is at least one neighbor, and the coin flips are independent, every early round has probability at least $1/4$ to be fruitful, regardless of the prior history.
Consider the first $8 C \log n$ early rounds. In expectation, at least $2 C \log n$ of them are fruitful. Applying Chernoff's bound for the lower tail, with $\delta = 1/2$, we have 
$$
    \Pr(\mbox{number of fruitful rounds} < C\log n) 
    < \exp\left(- \frac{(1/2)^2}{2} \cdot 2 C \log n\right) < \frac{1}{n^2},
$$
where the last inequality holds for large enough constant $C$.

Hence, with probability at least $1 - 1/n^2$, in the first $8 C \log n$ early rounds there are at least $C \log n$ fruitful rounds. But, since each fruitful round causes $v$ to drop out of the current Luby phase, and there are at most $C \log n$ Luby phases, it follows that $v$ is active for at most $8 C \log n$ early rounds. Adding these to the at most $\beta \log n$ late rounds (discussed earlier), and at most $C \log n$ rounds in between Luby phases, corresponding to lines 13 or 17 in the pseudocode, our final upper bound on energy spent is $(9 C + \beta) \log n = O(\log n)$.

A union bound over all vertices $v$ shows that, with  probability at least $1 - 1/n$, no vertex spends more than $O(\log n)$ energy, completing the proof.
\end{proof}

\begin{proof} [Proof of Lemma~\ref{lem:maintainsMIS}]
We prove by induction on $i \in [1, C \log n + 1]$ the following claim: at the start of any Luby phase $i$, the nodes with status \inMIS{} form an independent set with probability at least $1-9(i-1)/n^2$. (For $i = C \log n + 1$, you can consider instead the end of Luby phase $C \log n$.) 

For the base case of $i = 1$, the statement holds trivially and with probability 1 since no node has status \inMIS{}. Now, we consider some Luby phase $i \ge 1$, such that the nodes with status \inMIS{} at its start, denoted by $M_i$, is independent (with probability at least $1-9(i-1)/n^2$). Note that in Luby phase $i$, only nodes having joined $W_i \cup C_i$ during the competition may change their status to \inMIS{}. Out of these, nodes in $W_i$ form an independent set with probability at least $1-6/n^2$ (by Lemma \ref{lem:Wi-Independent} and a union bound over all nodes pairs). Moreover, by Lemma \ref{lem:ETRBackoffCorrectness} (and our choice of $C'$), any node in $W_i$ detects if it has a neighbor in $M_i$ with probability at least $1-1/n^5$, in which case it does not set its status to \inMIS{}. By a simple union bound, with probability at least $1-1/n^{4}$, all nodes in $W_i$ detects whether it has a neighbor in $M_i$. Then, if we let $W_i^*$ denote those nodes of $W_i$ that joined the MIS, then $M_i \cup W_i^*$ is independent with probability at least $1-6/n^2-1/n^4$.

Next, we consider the nodes of $C_i$ that change their status to \inMIS{}. Note that any node in $C_i$ detects if it has a neighbor in $M_i \cup W_i^*$ with probability at least $1-1/n^5$, in which case it does not set its status to \inMIS{}. Once again, we can apply a union bound, but now to the nodes of $C_i$. It remains to consider nodes in $C_i^*$ --- defined as the nodes of $C_i$ that did not detect any MIS neighbors --- because $C_i^*$ may contain some adjacent nodes. However, nodes in $C_i^*$ run \peterAlg{} on an induced subgraph of maximum degree $O(\log n)$ (by Corollary \ref{cor:goodDegreeEstimate}) and hence by applying Davies' algorithm (cf. Section \ref{subsec:DaviesAlg}),  nodes in $C_i^*$ that have status \inMIS{} (denoted here by $M_i'$) form an independent set with probability at least $1-1/n^2$.
Therefore, it follows that the set of all nodes that have status \inMIS{} by the end of Luby phase $i$ (and thus at the start of phase $i+1$), which is $M_i \cup W_i^* \cup M_i'$, is independent with probability at least $1 -9(i-1)/n^2 -  (6/n^2 + 2/n^{4}+1/n^2) \geq 1 -  i \cdot (9/n^2)$.
This completes the induction step, and the lemma statement follows when we consider $i \le C \log n + 1$.
\end{proof}

To prove the following two lemmas, we will need the following definitions. For any Luby phase $i$, let $D_i^{start}$ denote all non-terminated, MIS-dominated nodes at the start of the phase, and $D_i$ denote all such nodes immediately prior to the shallow check of that phase (i.e., in Line~\ref{line:shallowCheck} of Algorithm~\ref{alg:no-CD-MIS}).
Moreover, let $X_i$ denote the number of edges in the residual graph (i.e., in $E_i$) incident to $D_i$, and for any $v \in D_i$, let $X_i(v)$ denote the number of such edges incident to $v$. 

\begin{proof} [Proof of Lemma~\ref{lem:eliminate-edges-ideal}]
Let us denote by $N_i(v) = N(v) \cap V_{i-1}$ the neighbors of $v$ within the residual graph, and within those, by $N_D(v) = N_i(v) \cap D_i^{start}$ those that start the Luby phase as MIS-dominated, and by $N_F(v) = N_i(v) \setminus N_D(v)$ those that do not. 

Now, consider some arbitrary node $v$. First, note that any neighboring node $u \in N_F(v)$ that enters the MIS implies that $v$ is in $D_i$ (immediately prior to the shallow check) and thus that $X_i(v) \geq |N_i(v)|$. Second, if no neighbors of $v$ (nor $v$ itself) enters the MIS, then it still holds that $X_i(v) \geq |N_D(v)|$ since $D_i^{start} \subseteq D_i$ (because nodes in the MIS never change their status). Hence, we have 
$$ \ExpCond{X_i(v)}{E_{i-1}} \geq |N_D(v)| + \Prob{v \in D_i \mid E_{i-1}} |N_F(v)| $$

Next, we lower bound the probability that $v \in D_i$, conditioned on $E_{i-1}$. 
We say that some neighbor $u \in N_F(v)$ is eligible with respect to $v$, and denote as $E_u$ the corresponding event, if $u \in N_F(v)$ chooses a maximum bitstring over $N_i(v) \cup N_i(u)$ in the competition of this Luby phase. (For the sake of this analysis, we assume nodes with status \inMIS{} also choose a bitstring---contrary to the algorithm description---but that each such bitstring is smaller than any bitstring chosen by any node without an \inMIS{} status.) Note that $E_u$ implies that $u$ chose a locally maximum bitstring, thus by Lemma \ref{lem:Wi-notEmpty}, $u$ sets its status to \win{} after the competition of this Luby phase with probability at least $1 - 1/n^2$. Moreover, by definition of $D_i^{start}$, $u$ has no neighbors with status \inMIS{} during the competition and the first deep check, so by the algorithm description, 
$u$ sets its status to \inMIS{} in this Luby phase, prior to Line \ref{line:shallowCheck}. In summary, $E_u$ implies that $v$ is in $D_i$ with probability $1 - 1/n^2$. Moreover, note that the $E_u$ events are mutually exclusive over $N_F(v)$, and $\Pr(E_u \mid E_{i-1}) \geq 1/(|N_i(v)|+|N_i(u)|)$. Hence, we have that 

$$ \Pr(v \in D_i \mid E_{i-1}) \geq \sum_{u \in N_F(v)} \Prob{E_u \mid E_{i-1}} \left(1 - \frac{1}{n^2}\right) \geq \left(1 - \frac{1}{n^2}\right) \sum_{u \in N_F(v)} \frac{1}{|N_i(v)|+|N_i(u)|} $$

Finally, we shall lower bound the expectation of $X_i$ conditioned on $E_{i-1}$. By linearity of conditional expectation, and because we count each edge of $X_i$ twice when summing $X_i(v)$ over all nodes $v \in V_{i-1}$, we have that

\begin{align*}
\ExpCond{X_i}{E_{i-1}} 
&= \frac{1}{2} \sum_{v \in V_{i-1}} \ExpCond{X_i(v)}{E_{i-1}} \\
&\geq \frac{1}{2} \sum_{v \in V_{i-1}} |N_D(v)| + \frac{1}{2} \left(1 - \frac{1}{n^2}\right) \sum_{v \in V_{i-1}} \sum_{u \in N_F(v)} \frac{|N_F(v)|}{(|N_i(v)|+|N_i(u)|} 
\end{align*}

Let us denote by $E_{i-1}^D$ all edges of $E_{i-1}$ with at least one endpoint in $D_i$. Then,  
$$\frac{1}{2} \sum_{v \in V_{i-1}} |N_D(v)| \geq |E_{i-1}^D|/4 + \sum_{v \in V_{i-1}} |N_D(v)| /4$$

Next, let us denote by $V_{i-1}^+$ all nodes in $V_{i-1}$ for which $|N_D(v)| \geq |N_i(v)|/2$, and by $E_{i-1}^+$ all edges of $E_{i-1}$ with at least one endpoint in $V_{i-1}^+$. 
Then, 

$$\sum_{v \in V_{i-1}} |N_D(v)| /4 \geq  \sum_{v \in V_{i-1}^+} |N_i(v)| /8 \geq |E_{i-1}^+| / 8$$

Finally, if we define $E_{i-1}^R = E_{i-1} \setminus (E_{i-1}^+ \cup E_{i-1}^D)$, then we can reorder the double sum and ignore some terms to get, for the last term, that

\begin{align*}
\frac{1}{2} \left(1 - \frac{1}{n^2}\right) \sum_{v \in V_{i-1}} \sum_{u \in N_F(v)} \frac{|N_F(v)|}{(|N_i(v)|+|N_i(u)|} 
&\geq \frac{1}{2} \left(1 - \frac{1}{n^2}\right) \sum_{\{u,v\} \in E_{i-1}^R} \frac{|N_F(u)| + |N_F(v)|}{|N_i(v)|+|N_i(u)|} \\
&\geq \frac{1}{4} \left(1 - \frac{1}{n^2}\right) |E_{i-1}^R| \geq |E_{i-1}^R|/5
\end{align*}

where the last inequality holds for large enough $n$.
It follows that  $$\ExpCond{X_i}{E_{i-1}} \geq |E_{i-1}| /8.  \qedhere
$$
\end{proof}

\begin{proof} [Proof of Lemma~\ref{lem:eliminate-edges}]
Consider any Luby phase $i$. By Lemma~\ref{lem:ETRBackoffCorrectness}, every node in $D_i$ detects the presence of a neighbor with status \inMIS{} with probability at least $1/8$. If that happens, then that node sets its status to \notInMIS{} (and terminates) by the end of the Luby phase and thus every of its incident edge leaves the residual graph. In other words, every edge incident on $D_i$ leaves the residual graph (i.e., is not in $E_i$) with probability at least $1/8$. There are $X_i$ such edges, and by Lemma~\ref{lem:eliminate-edges-ideal}, $\ExpCond{X_i}{E_{i-1}} \ge \frac{|E_{i-1}|}{8}$.
Hence, it follows that $\ExpCond{|E_i|}{E_{i-1}} \leq  (1-\frac{1}{64}) |E_{i-1}|$.
\end{proof}

\section{Energy-Efficient Backoff Procedures}
\label{app:backoff}

Designing algorithms in the no-CD model can be significantly more difficult than in the CD model. In particular, a crucial difference is that in the no-CD model, nodes can no longer distinguish silence from collisions. 
Hence, nodes work with less information than they otherwise would have. In particular, the only way for a node to determine whether one of its neighbors is sending is for \emph{exactly one} of its neighbors to send while it listens.

A generic way to achieve this is via exponential backoff. (This is a well-known procedure, and is also referred to as \Call{Decay}{} in some works.) 
At a high-level, nodes decide to take either a \emph{sender} or \emph{receiver} role for the entire backoff protocol, and the protocol ensures that any receiver that has at least one sender neighbor hears a message with constant probability~\cite{BGI92}.  
More concretely, traditional backoff works as follows. In a first round, all sender nodes send a message, while receiving nodes simply listen. Then, each sender node flips a fair coin to decide whether to send again in the next round (say when flipping 1) or drop out of the backoff (when flipping 0). This repeats for $O(\log \Delta)$ rounds.
These iterations of $O(\log \Delta)$ rounds can themselves be repeated, say up to $k$ times, to boost the success probability --- this follows from well-known statements on exponential backoff, but see also Lemma \ref{lem:ETRBackoffCorrectness}. And if we take $k$ large enough, say $\Theta(\log n)$, then the success is guaranteed with high probability.

In the above (traditional) exponential backoff, all nodes must be awake in all $O(\log n \log \Delta)$ rounds. In contrast, we give energy-efficient adaptations of the traditional exponential backoff procedures. The sender-side backoff is modified so that senders transmit only once per iteration, leading to a guaranteed and significant energy efficiency for senders. The receiver-side backoff is modified so that once a node hears a message, it sleeps for the remainder of the backoff (i.e., essentially an energy-motivated early ``termination''). Note that any receiver node with no sender neighbor will be awake throughout the entire backoff, while any receiver node with at least one sender neighbor will save a significant amount of energy in expectation. More concretely, in the latter case, such a receiver node will be awake in expectation for only a constant number of iterations before it hears a message and sleeps.


\begin{algorithm}[ht]
\caption{Energy-efficient $k$-Repeated Backoff Procedures}\label{alg:}
\begin{algorithmic}[1]
\Procedure{\SndBackoff{}}{$k$, $\Delta$}
\LComment{Senders send once per iteration, and any listener hears a sender neighbor (if one exists) with constant probability per iteration.}
\For{$i \gets 1$ to $k$}
    \State $x \gets \text{Sample from a geometric distribution with parameter } \frac{1}{2}$
    \State $x \gets \min(x, \lceil\log \Delta\rceil)$
    \For{$j \gets 1$ to $\lceil\log \Delta\rceil$}
        \If{$j = x$} 
            \State \transmit{} 1
        \Else
            \State \sleep{}
        \EndIf
    \EndFor
\EndFor
\EndProcedure

\\

\Procedure{\RecBackoff{}}{$k$, $\Delta$, $\DeltaEst = \Delta$}
\LComment{The third argument is optional: when not specified, it defaults to $\Delta$.}
\State $heard \gets$ False
\For{$i \gets 1$ to $k$}
    \LComment{While they have not yet heard a message, receivers listen for log of their approximate degree rounds per iteration.}
    \For{$j \gets 1$ to $\lceil\log \Delta\rceil$}
        \If{not $heard$ and $j \leq \lceil\log  \DeltaEst \rceil$} 
            \State \listen{}
            \If{heard 1}
                \State $heard \gets$ True
            \EndIf
        \Else
            \State \sleep{}
        \EndIf
    \EndFor
\EndFor
\State \Return $heard$
\EndProcedure
\end{algorithmic}
\end{algorithm}

Note that senders and receivers have asymmetric energy complexities (captured in the below lemma), that is, senders use a logarithmic factor less energy than receivers. This asymmetry is a crucial factor in the low energy complexity of our MIS algorithm in Section \ref{sec:noCDMIS}. These energy-efficient backoffs also provide the same correctness guarantees as the non-energy-efficient ones.

\begin{proof} [Proof of Lemma~\ref{lem:ETRBackoffUpperBounds}]
    The round complexity upper bounds follow from the fact that both backoff procedures execute $k$ iterations, each taking $O(\log \Delta)$ rounds. As for the awake complexity upper bounds, they follow from the fact that nodes transmit once per outer loop iteration in $\Call{\SndBackoff{}}{}$ and listen at most $O(\log \DeltaEst)$ times per outer loop iteration in $\Call{\RecBackoff{}}{}$.
\end{proof}

\begin{proof} [Proof of Lemma~\ref{lem:ETRBackoffCorrectness}]
    Consider a receiver node $v$, and at most $\DeltaEst$ sender nodes neighboring $v$. First, note that these sender nodes participate in all backoff iterations, thus any receiver node has the same number of sender neighbors throughout all backoff iterations. 
    
    Moreover, the following claim holds for each backoff iteration:
    if (receiver) node $v$ has at least one sender neighbor, then during that iteration $v$ hears a message with probability at least 1/8. Indeed, for any single backoff iteration, each sender chooses to transmit in round $j < \lceil \log \Delta \rceil$ with probability $1/2^j$, and in round $\lceil \log \Delta \rceil$ with probability $1/2^{\lceil \log \Delta \rceil - 1}$ (due to the capping). Let $2 \leq d_S(v) \leq \DeltaEst$ be the number of sender neighbors of $v$. Note that if $d_S(v) = 1$, the lemma holds trivially. Then, in round $j = \lceil \log d_S(v) \rceil$, for which $v$ is awake and $1/2^j \in [1/(2d_S(v)),1/d_S(v)]$, the probability that there is exactly one sender neighbor of $v$ transmitting in round $j$ is (by summing the probabilities of the mutually exclusive events that a given sender neighbor transmits alone) at least 
    
    $$ \sum_{c=1}^{d_S(v)} \frac{1}{2^j} \left(1-\frac{1}{2^j}\right)^{d_S(v)-1} \geq d_S(v) \cdot \frac{1}{2 d_S(v)} \left(1-\frac{1}{d_S(v)}\right)^{d_S(v)-1} \geq \frac{1}{8} $$
    where the last inequality follows from $1-x \geq (1/4)^x$ for $x \in [0,\frac{1}{2}]$. Since $v$ listens in all rounds (unless it has already heard a message), $v$ hears a message in round $j$ or earlier. 
    
    Finally, since the randomness used by the successive backoff iterations are independent, it follows that $v$ does not hear any of its sender neighbor in $k$ backoff iterations with probability at most $(7/8)^{k}$, or equivalently, learns that it has at least one sender neighbor with probability at least $1 - (7/8)^{k}$. 
\end{proof}

\end{document}